\documentclass{article}

\usepackage[latin1]{inputenc}  
\usepackage{amsmath,amsfonts,amsthm,amssymb}  
\usepackage{dsfont}			
\usepackage{thmtools}		
\usepackage{enumerate}
\usepackage{graphicx}			

\usepackage{subcaption}			

\usepackage[latin1]{inputenc}                       

\newtheorem{theorem}{Theorem}[section]
\newtheorem{proposition}[theorem]{Proposition}

\newtheorem*{postulate}{Postulate}

\theoremstyle{definition}
\newtheorem*{definition}{Definition} 

\newtheorem*{example}{Example}

\declaretheoremstyle[
spaceabove=6pt, spacebelow=6pt,
headfont=\normalfont\bfseries,
notefont=\normalfont\bfseries, 
notebraces={}{},
bodyfont=\normalfont\itshape
]{Estilo1}

\declaretheorem[style=Estilo1,numbered=no,name=\!\!]{ThmName} 

\newcommand{\Schrodinger}			{Schr\"o\-din\-ger}

\def\im	 				{\mathrm{i}}

\newcommand\ket[1]					{\left| #1 \right\rangle}
\newcommand\braket[2] 			{\left\langle #1 | #2 \right\rangle}
\newcommand\norm[1]					{\left\| #1 \right\|}

\usepackage{mathtools} 		
\DeclarePairedDelimiter{\Prod}{\langle}{\rangle}  

\newcommand\R               {\mathds{R}}
\newcommand\C               {\mathds{C}}                    
\newcommand\PP               {\mathds{P}}

\begin{document}

\title{Quantum fractionalism: the Born rule as a consequence of the complex Pythagorean theorem} 

\author{Andr\'e L. G. Mandolesi 
               \thanks{Instituto de Matemática e Estatística, Universidade Federal da Bahia, Av. Adhemar de Barros s/n, 40170-110, Salvador - BA, Brazil. E-mail: \texttt{andre.mandolesi@ufba.br} 
               }}

\date{\today}

\maketitle

\abstract{Everettian Quantum Mechanics, or the Many Worlds Interpretation, lacks an explanation for quantum probabilities. We show that the values given by the Born rule equal projection factors, describing the contraction of Lebesgue measures in orthogonal projections from the complex line of a quantum state to eigenspaces of an observable. Unit total probability corresponds to a complex Pythagorean theorem: the measure of a subset of the complex line is the sum of the measures of its projections on all eigenspaces.

To show that projection factors can work as probabilities, we postulate the existence of a continuum infinity of identical quantum universes, all with the same quasi-classical worlds.
In a measurement, these factors give the relative amounts of worlds with each result, which we associate to frequentist and Bayesian  probabilities. This solves the probability problem of Everett's theory, allowing its preferred basis problem to be solved as well, and may help settle questions about the nature of probability.
} 

\vspace{.5em}

\section{Introduction}

Since the early days of Quantum Mechanics, its probabilistic nature has baffled many physicists, most notoriously Einstein. This motivated alternatives to the Copenhagen interpretation (CQM), like hidden-variables theories, which remain unsuccessful \cite{Auletta2000,Home1997,Wheeler2014}. Today most physicists accept the theory as intrinsically probabilistic, but problems remain, regarding quantum measurements and the quantum-classical transition. They have gained new relevance as experimental and theoretical developments, such as quantum computation and information \cite{Nielsen2010}, have pushed the limits of the theory to increasingly larger scales.

An alternative to CQM which is particularly well suited to be applied to macroscopic systems is Everettian Quantum Mechanics (EQM), or the Many Worlds Interpretation \cite{EverettIII1957,DeWitt1973}. It rejects the measurement postulate and describes quantum measurements in terms of entanglement: all results happen, but entangled to different states of the observer, in a superposition of distinct worlds. But as EQM is fully deterministic, it faces the problem of explaining the probabilities observed in quantum experiments.

Many attempts have been made to get the Born rule from more fundamental principles. The most famous is Gleason's theorem \cite{Gleason1957}, but it relies on a hypothesis which can not be justified before we know how probabilities can emerge in EQM. Everettian Decision Theory (EDT) \cite{Deutsch1999,Saunders2010a,Wallace2012} has gained some acceptance among Everettians, but it has many problems \cite{Dawid2015,Mandolesi2019}.

To get probabilities in EQM, we first note that they have the same values as projection factors, which describe how Lebesgue measures in a ray (the complex line of a quantum state, minus the origin) contract when orthogonally projected onto eigenspaces of an observable. A complex Pythagorean theorem \cite{Mandolesi_Pythagorean} says the measure of a subset of a ray equals the total measure of its projections on all eigenspaces, corresponding to the condition of unit total probability.  

This is not enough to ensure projection factors can work as probabilities, and we add a new postulate to EQM.
This goes against a certain purist view of the theory, focused on the promise of deriving all of quantum mechanics from its unitary part. We, however, see EQM as a step towards a coherent quantum theory, whose development should not be constrained by a lesser goal.

We postulate the existence of a continuum infinity of universes, but different from that proposed by Deutsch \cite{deutsch1985quantum,Deutsch1997}, as our universes are always identical, corresponding to all states in a ray, and evolve deterministically in the same way. When a quantum experiment is carried out in one of them, it is also performed in the others, producing the same worlds in all of them, and in the end we have a continuum infinity of worlds for each result. 

In terms of Lebesgue measures, relative amounts of worlds with each result are shown to equal the corresponding projection factors, so they have the same values as the Born rule probabilities.
In nearly all worlds they correspond to observed frequencies, and can be used for decisions and inferences, so this solves the probability problem. And as our method does not rely on a previous solution of EQM's preferred basis problem, it removes an obstacle \cite{Dawid2015} to Wallace's decoherence based solution of this other problem \cite{Wallace2012}.

\emph{Quantum Fractionalism} is this interpretation of quantum probabilities as relative amounts (or fractions) of worlds resulting from a quantum measurement. It provides, at least for such cases, an arguably better concept of probability than Frequentism or Bayesianism. These theories give good accounts of how probabilities work, but fail to provide fully satisfactory definitions of what they are, a notoriously difficult problem in the philosophy of probability \cite{Gillies2000,Mellor2005}. If classical probabilities can be traced back to quantum origins \cite{Albrecht2014}, our interpretation might help clarify the nature of probability in general. 

Getting probabilities as world measures is not a new idea. The problem is justifying and interpreting the measure: does it give a notion of amount or size for sets of worlds, or is it something else? Is small measure linked in some sense to negligibility?
Everett \cite{EverettIII1957} tried attributing a measure to worlds, but its meaning was not clear.
Deutsch's continuum of universes \cite{deutsch1985quantum} was postulated to have a size measure, with values chosen to give the desired result. 
Greaves'  version of EDT \cite{Greaves2004} uses `caring measures', that show how much we should care about our future selves, but do not represent amounts of worlds. 
Araújo \cite{araujo2019probability} links EDT's subjective probabilities to objective ones described by a size measure given by a $p$-norm, but uses empirical results to justify $p=2$.

Our measure has the advantages of having a concrete interpretation (relative amounts in sets of worlds), and that its values are obtained in a simple way from the geometry of a complex Hilbert space, through a natural correspondence between worlds and their quantum states.
Our method also shows that in non-complex Hilbert spaces, whose use has been proposed by some authors \cite{Stueckelberg1960a,Finkelstein1962,Adler1995}, an Everettian Born rule would disagree with the observed values. 

Section \ref{sec:Preliminaries} reviews CQM, EQM and probability theories. 
Projection factors and the complex Pythagorean theorem are described in section \ref{sc:projection factors}.
In section \ref{sc:universes and worlds} we present our postulate, our universes and worlds, link projection factors to fractions of worlds, and discuss non-complex Hilbert spaces.
In section \ref{sc:Quantum Fractionalism} we show these fractions can be interpreted as probabilities, and discuss the epistemic problem. 
Section \ref{sc:Deutsch} compares our approach to Deutsch's.
In section \ref{sc:Physical hypotheses} we discuss our postulate, what we can infer from it, and whether it might follow from more basic principles.
We conclude with some remarks in section \ref{sc:Conclusion}.

\section{Preliminaries}\label{sec:Preliminaries}

In this section we present EQM in its modern form, which brings decoherence into the fold, the problems of CQM that motivated it, and the new ones it brings. We also briefly review some theories of probability and their difficulties.

\subsection{Problems of the Copenhagen interpretation}

Despite its experimental success and wide acceptance, CQM has well known theoretical problems, and the literature on the subject is quite extensive (see, for example, \cite{Auletta2000,Home1997,Wheeler2014} and references therein).

One is the \emph{measurement problem}: the measurement postulate reflects accurately what is observed in experiments, but is ambiguous in terms of when it should come into play. It sets measurements apart from other quantum processes,  as only in them unitary evolution governed by \Schrodinger's equation gives way to the probabilistic collapse of the quantum state. But there is no clear definition of which characteristics a process must have to trigger such change and count as a measurement. 

If the microscopic system, measuring device and observer are collections of particles interacting with each other according to quantum laws, there ought to be an explanation of what switches their evolution from one governed by a deterministic equation to a probabilistic one. Or at least a clear rule specifying when this happens, without recourse to high level terms like ``measurement'' (how is an electron to `know' the particles interacting with it belong to a measuring device and so it should stop obeying \Schrodinger's equation?). 

Many answers as to what constitutes a quantum measurement have been proposed, each raising even more doubts.
If measurements consist in reducing ignorance about a quantum state, must they involve conscious beings or can a device cause the wavefunction to collapse? If it is only a matter of information transferring between systems, a particle being affected by the state of another should count as a measurement, so why does it not cause a colapse?
If it is an issue of a macroscopic system interacting with a microscopic one, how big is macroscopic enough, and why is size relevant if only close particles interact significantly?
If it involves information spreading irreversibly, when is the limit of reversibility reached, and how does it trigger the change in quantum behavior? 
Also, how exactly does the collapse happen?

This problem is connected to that of the \emph{quantum-classical transition}. What is the range of applicability of Quantum Mechanics? It works for microscopic systems, but what happens as the number of particles grows? Does it gradually turn into Classical Mechanics, as is the case with Relativity as velocities decrease? Some quantum relations become classical ones if we take averages and let $\hbar\rightarrow 0$, but not everything transitions well. In the usual view, macroscopic quantum superpositions should not happen, lest we have \Schrodinger\ cats, but it is not clear what might eliminate them as systems get bigger. Decoherence has been suggested as a solution \cite{Joos2003,Schlosshauer2007,Zurek2002}, but even if it eliminates interference between components of a macroscopic superposition it does not explain the disappearance of all but one of them. This problem has led some to consider Quantum Mechanics valid only for microscopic systems, with a more universal theory being needed to connect quantum and classical physics. 

For a long time most physicists have brushed aside such difficulties, believing that these were philosophical questions of little physical relevance, that they had been settled in the famous Bohr-Einstein debate, or that these were minor flaws in an otherwise very precise theory, which would end up being fixed. But as quantum theory reaches its centenary the problems remain, and gain increasing relevance as new theoretical and experimental advances allow us to explore the limits of the theory in ways that were inconceivable a few decades ago.

\subsection{Everettian Quantum Mechanics}\label{sec:EQM}

In 1957, H. Everett III \cite{EverettIII1957} took a fresh look at what would happen if quantum theory was applied to macroscopic systems and measurements were regular quantum processes, proposing what became known as Everettian Quantum Mechanics, or the Many Worlds Interpretation \cite{DeWitt1973}. 

Before detailing and justifying his ideas, we give a general view. 
In EQM, quantum theory is universally valid, for small and large systems alike, but the measurement postulate is rejected, with all systems evolving deterministically at all times according to \Schrodinger's equation.  This  leads to macroscopic superpositions, but does not contradict our classical experience, as EQM explains why we do not perceive them. 
Quantum measurements are basically entanglements at a macroscopic scale, and split the quantum state of the Universe into components called branches or worlds, with all possible results happening in some of them. There is no collapse of the quantum state, but when an observer interacts with the outcome of an experiment he splits into different versions of himself, each seeing only the result of his branch, as if the collapse had happened. 

This may seem like an almost mystical departure from conventional quantum theory, but Everettians see it as a direct consequence of unitary quantum mechanics (the usual formalism minus the measurement postulate) applied to macroscopic systems: linearity of Schr\"{o}dinger's equation and algebraic properties of the tensor product lead naturally to branching. They claim all of quantum mechanics follows from its unitary part, with no need for new postulates.

Let us detail the branching process. In EQM, a measurement is just quantum entanglement of the measuring device with whatever is being measured. More precisely, a \emph{measuring device} for a basis $\{\ket{i}\}$ of a system is any apparatus, in a quantum state $\ket{D}$, which interacts in such a way that its state becomes correlated with elements of this basis, i.e. if the system is in state $\ket{i}$ the composite state evolves as
\begin{equation*}
\ket{i}\otimes\ket{D}\ \  \longmapsto\ \  \ket{i}\otimes\ket{D_i},
\end{equation*}
where $\ket{D_i}$ is a new state of the device, registering result $i$ (or simply reflecting in some way the fact that it interacted with $\ket{i}$).
Linearity of Schr\"odinger's equation implies that, if the system is in a superposition $\ket{\Psi}=\sum_i c_i\ket{i}$, the composite state evolves as
\begin{equation}\label{eq:device branches}
\ket{\Psi}\otimes\ket{D} = \Big(\sum_i c_i\ket{i}\Big)\otimes\ket{D} \ \  \longmapsto\ \   \sum_i c_i\ket{i}\otimes\ket{D_i}.
\end{equation}

This final state is to be accepted as an actual quantum superposition of macroscopic states. But it will not be perceived as such by an observer looking at the device, as, by the same argument, his state $\ket{O}$ will evolve into a superposition,
\begin{equation}\label{eq:observer branches}
\Big(\sum_i c_i\ket{i}\otimes\ket{D_i}\Big)\otimes\ket{O} \ \  \longmapsto\ \   \sum_i c_i\ket{i}\otimes\ket{D_i}\otimes\ket{O_i},
\end{equation}
with $\ket{O_i}$ being a state in which he saw result $i$. The interpretation is that he has  split into different versions of himself, each seeing a result. 

By linearity, the components $\ket{i}\otimes\ket{D_i}\otimes\ket{O_i}$ evolve independently, and interference is negligible if they are distinct enough, as tends to be the case with macroscopic systems. 
Each $\ket{O_i}$ evolves as if the initial state had been $\ket{i}\otimes\ket{D}\otimes\ket{O}$, so he does not feel the splitting or the existence of his other versions, and for him it is as if the system had collapsed to $\ket{i}$.
Each component is called a \emph{world} or \emph{branch}, and this evolution of one world into a superposition of many is called \emph{branching}. So in EQM all possible results of a measurement do happen, but in different worlds. 

Some physicists reject such `splitting of worlds' for seeing it as a drastic process that creates whole new worlds, with new particles and energy. But it is just an ordinary evolution of a state into a quantum superposition. An electron with spin up can evolve into a superposition of spins up and down, with the new component appearing smoothly as its coefficient increases from 0, and this is not seen as the creation of a new electron with spin down. Likewise, branchings are normal quantum processes in which the state of the Universe gradually evolves into a superposition of macroscopically distinct states of the same Universe.


The description of EQM presented here is somewhat incomplete, and neo-Everettians \cite{Saunders2010a} include decoherence as an essential ingredient. The reason is that, while problems of CQM disappear in EQM, new ones come along, such as the preferred basis and probability problems.

\subsubsection{Preferred basis problem}\label{sec:Preferred_Basis_Problem}

This problem is how to find a natural way to decompose a macroscopic quantum state into branches behaving like the classical reality we observe (even if not all of them, and not all the time).

Decoherence is seen as a mechanism through which a (quasi-)classical world might emerge from a quantum universe \cite{Joos2003,Schlosshauer2007,Zurek2002}. Wallace \cite{Saunders2010a,Wallace2012} has proposed a solution to the problem by combining EQM with the decoherent histories formalism \cite{Gell-Mann1990,Gell-Mann1993}. But this requires solving first the probability problem, as the probabilistic interpretation is used to justify approximations in decoherence  \cite{Baker2007,Zurek2005}. 
Wallace \cite[pp.\,253--254]{Wallace2012} claims the Hilbert space norm measures significance even without probabilities, but this has been contested \cite{Dawid2015,Mandolesi2017}.

Still, it seems plausible that a branch decomposition, if it can be obtained, should be along these lines.
An important characteristic of such decoherence based approach is that the decomposition is not clear-cut or unique. In EQM measurements lose their special status, becoming just interactions that produce entanglement. But particles get entangled all the time, and if each such process counts as a measurement of a particle by another, branching becomes a pervasive phenomenon, with worlds splitting all the time into a myriad of others. And lots of them will be nearly identical, differing only in the states of a few particles, which is not enough to ensure they evolve with negligible interference.

As Wallace \cite{Wallace2012}  argues, a coarse-graining of similar branches can solve these difficulties, bringing more stability to the decomposition and ensuring the resulting worlds are distinct enough to have negligible interference. This, however, introduces some arbitrariness into the decomposition, as the resulting worlds would depend on the chosen fineness of grain.

\subsubsection{Probability problem}\label{sec:Probability_Problem}

In CQM, the connection between theory and experiment is provided by the measurement postulate, including the Born rule.

\begin{ThmName}[Born Rule]
In CQM, the probability of obtaining result $i$ when measuring  a normalized state $\ket{\Psi}=\sum_i c_i\ket{i}$ in an orthonormal basis $\{\ket{i}\}$ is   
\begin{equation}\label{eq:Born rule}
p_{\Psi,i}=|c_i|^2=|\!\braket{i}{\Psi}\!|^2.
\end{equation}
\end{ThmName}

Normalized states are convenient to simplify some formulas, but they are not physically special in any way. 
Claims that normalization is necessary for unit total probability are easily circumvented: without it, we can write 
$\Psi = \sum_{i} \psi_i$, where $\psi_i=c_i\ket{i}$, and replace \eqref{eq:Born rule} with
\begin{equation}\label{eq:Born unnormalized}
p_{\Psi,i} =\frac{|\!\braket{\psi_i}{\Psi}\!|^2}{{\norm{\psi_i}^2\norm{\Psi}^2}} = \frac{\|\psi_i\|^2}{\|\Psi\|^2}. 
\end{equation}

In EQM, any result with $c_i \neq 0$ is obtained with certainty, in some world(s). The \emph{probability problem} is to reconcile this with experiments, which suggest that results are probabilistic and follow the Born rule.
In this context, one usually refers to the values in \eqref{eq:Born rule} as \emph{Born weights}.

The problem has a qualitative aspect: how can probabilities emerge in a deterministic theory? In Classical Mechanics processes can appear random due to ignorance of details, but for EQM we need probabilities which apply even if the quantum state and its evolution are perfectly known. 
Wallace \cite{Wallace2012} defends an operational and functional definition of probability, via decision theory and Bayesian inference.
Vaidman \cite{Vaidman1998} claims there is a \emph{self-locating uncertainty} in the time between processes \eqref{eq:device branches} and \eqref{eq:observer branches}, as branching has already happened, but each version of $\ket{O}$ is still ignorant as to which branch he is in. 

There is also a quantitative aspect of accounting for probability values. 
The idea that a measurement with $n$ results produces $n$ branches seems natural, leading to probability $\frac{1}{n}$ for all results, in disagreement with experiments.
Wallace \cite{Wallace2012} dismisses this by saying there is no good way to count branches: the final state can have more than $n$ branches, as unpredictable extraneous interactions can cause additional branchings, with an uneven and shifting distribution among the results. 
Coarse-graining reduces and stabilizes the number of branches, but makes it somewhat arbitrary, depending on the chosen fineness of grain.

Everett \cite{EverettIII1957} proved that if a measure can be attributed to branches, and satisfies some hypotheses, like being preserved under finer decompositions, it must agree with the values in \eqref{eq:Born rule}. And, as the number of measurements tends to infinity, the total measure of  branches deviating from the Born rule tends to $0$. But for finite experiments this only means branches deviating beyond a given error have small measure, and without an interpretation (probabilistic or otherwise) for such measure this does not make them any less relevant.
The same difficulty affects a similar idea by Graham \cite{Graham1973} and Wallace's use of decoherence to solve the preferred basis problem, which is a prerequisite for his decision theoretic proof of the Born rule \cite{Wallace2012}. 

Gleason \cite{Gleason1957} also obtains the Born rule, assuming the probability of a branch does not depend on the other branches: if a state decomposes as $\psi=\psi_1+\psi_2$ or $\psi=\psi_1+\psi_3+\psi_4$, depending on the orthogonal basis chosen, the probability for $\psi_1$ is the same in both. This hypothesis may seem reasonable if one has the Born rule in mind, but it is not natural for probability measures in general. It is violated by a counting measure, for example, and in section \ref{sc:real Hilbert space} we show it might not hold for quantum probabilities in non-complex Hilbert spaces.

Not knowing how probabilities can emerge in EQM, we can not assume they satisfy Everett's or Gleason's hypotheses. 
Other attempts \cite{Zurek2005,Albert1988,Buniy2006} have been made to get the right probabilities, without much success. 

Everettian Decision Theory \cite{Deutsch1999,Saunders2010a,Wallace2012}, which tries to get the Born rule in terms of subjective probabilities, has gained some acceptance, but faces many difficulties \cite{Dawid2015,Mandolesi2019}.
Roughly, it says it is rational to care about future branches in proportion to their Born weights. An agent should accept a bet in which $\frac{\sqrt{3}}{2}\ket{\uparrow}+\frac{1}{2}\ket{\downarrow}$ is measured, she gains \$1000,00 in up branches, and loses \$2000,00 in the rest, as the expected utility is $0.75\cdot 1000-0.25\cdot 2000>0$. For simplicity, we have equated utility with monetary value, but an utility function on rewards can be used for flexibility. 
It is important to note that the rational behavior prescribed by EDT only applies to pre-measurement decisions. It does not tell a post-measurement version of the agent who lost \$2000,00 that rationally she should be glad her Born weight is only $0.25$. She can not even take solace in thinking there are 3 times as many versions of herself who won \$1000,00, as Born weights are not interpreted in terms of amounts of branches.

\subsection{Philosophies of probability}

Even if we have, in most cases, an intuitive understanding of probabilistic statements, it is notoriously difficult to describe precisely what probability is.
Formally, it is a function in an event space satisfying Kolmogorov's axioms, but it still needs to be linked to how probabilities are used in daily life or science (areas of regions in a unit square satisfy the axioms, but are not probabilities per se).
Various theories of probability are described in \cite{Gillies2000,Mellor2005}, and some are discussed in connection with EQM  in \cite{Wallace2012}. Here we briefly sketch the main interpretations and their problems.

The simplest interpretation of probability may be the frequentist one. In it, saying each result of a fair die has probability $\frac{1}{6}$ means that, as the die is cast an increasingly large number of times, the relative frequency of each result will tend to $\frac{1}{6}$. 
But making precise sense of this is trickier than it seems. Even if we were to throw the die 6 million times, we can not say each result will occur 1 million times, give or take a thousand, only that it is highly probable this will happen. So even if relative frequencies can be used to measure probabilities, there is a circularity in using them to define probability.
Also, this depends on the possibility of repeating a test an arbitrary number of times, under equal conditions, and a die thrown 6 million times may get damaged along the way, breaking its symmetry. And frequentism does not apply to one time events: it can not make sense of the probability that a given candidate will win next election, or that stock prices will rise tomorrow.

Bayesianism is another interpretation, in which, roughly speaking, if Alice says the probability of getting result 6 in a die is $\frac{1}{6}$, it only means she would not be willing to bet on it at odds worse than 5 to 1. Bob, believing the die is loaded, might accept worse odds, and for him the probability would be different. Instead of being an objective concept, in this view probability only describes the credence one has about something, based on the information he has. As more data becomes available, a process of Bayesian updating allows the values of such subjective probabilities to be properly adjusted.

This view is useful when frequentism fails, as in the election or stock market examples, but it is harder to apply when probabilities appear to have a more objective nature, as in the decay of an atom. In truth, the distinction between objective and subjective probabilities is not as irreconcilable as it seems. It can be argued that, under certain rational constraints, and with enough information, they should not differ significantly.
Still, it is odd to define a decay probability in terms of one's willingness to bet on it. If he decides based on what Physics says are the objective probabilities, it brings back the problem of what these mean. If we define objective probabilities to be the values to which his subjective ones converge through Bayesian updating, as he learns from a large number of trials, it also leads to the question of what it is that he is learning about.
If Bayesian probabilities are subjective estimates based on incomplete information, what are they estimating, objective probabilities?

As discussed, in EQM a crude frequentist attempt to get probabilities by world counting gives wrong values, but fortunately it can not work as the number of worlds is ill defined.
Wallace \cite{Wallace2012} claims Bayesianism works better in EQM than classically, with rational constraints forcing subjective probabilities to agree with the Born rule, but he relies on dubious axioms \cite{Mandolesi2019}.
As we will show, EQM may indeed provide the best setting for a good probability concept, but it will be through other means.

\section{Complex Pythagorean Theorem}\label{sc:projection factors}

Some generalizations of the Pythagorean theorem involve areas or higher dimensional measures, usually relating the squared measure of a set to the squared measures of its orthogonal projections  on certain families of subspaces. But in certain versions for complex spaces the measures are not squared, due to dimensional reasons. We present a particular case, referring to  \cite{Mandolesi_Pythagorean} for details and more general formulations.

Let $H$ be a complex Hilbert space, with Hermitian product  $\Prod{\cdot,\cdot}$.
Complex spaces have underlying real ones, with twice the dimension, and the real part of $\Prod{\cdot,\cdot}$ gives a real inner product. As $\C$-orthogonality (with respect to $\Prod{\cdot,\cdot}$) implies $\R$-orthogonality (with respect to $\operatorname{Re}\Prod{\cdot,\cdot}$), orthogonal projections onto  complex subspaces coincide for both products.

A \emph{complex line} is a complex subspace $L\subset H$ with $\dim_\C L=1$, and the complex line of a nonzero $v\in H$ is $\C v=\{cv:c\in\C\}$.
It is isometric to a real plane, so we measure its subsets using the 2-dimensional Lebesgue measure (roughly speaking, the area), which we denote by $\left|\cdot\right|$.

\begin{definition}
The \emph{projection factor} of a complex line $L\subset H$ on a complex subspace $W\subset H$ is
\begin{equation*}\label{eq:projection factor}
\pi_{L,W}=\frac{|P(U)|}{|U|}, 
\end{equation*}
where $U\subset L$ is any Lebesgue measurable subset with $0<|U|<\infty$, and $P:L\rightarrow W$ is the orthogonal projection.
\end{definition}

As $P$ is linear, $\pi_{L,W}$ does not depend on the choice of $U$.

\begin{proposition}\label{pr:Pv pi}
Given $v\in H$ and a complex subspace $W\subset H$, let $P:\C v\rightarrow W$ be the orthogonal projection. Then
\begin{equation}\label{eq:Pv projection factor}
\|Pv\|^2=\|v\|^2\cdot\pi_{\C v,W}.
\end{equation}
\end{proposition}
\begin{proof}
As $P$ is $\C$-linear and $W$ is complex, the square $U\subset \C v$ of sides $v$ and $\im v$ projects to the square $P(U)\subset W$ of sides $Pv$ and $\im Pv$.
\end{proof}

\begin{definition}
An \emph{orthogonal partition} of $H$ is a collection $\{W_i\}$ of mutually orthogonal closed complex subspaces such that $H=\bigoplus_i W_i$.
\end{definition}

\begin{theorem}[Complex Pythagorean Theorem]\label{th:complex Pythagorean}
Given a nonzero $v\in H$ and an orthogonal partition $H=\bigoplus_i W_i$,
\begin{equation}\label{eq:sum factors}
\sum_{i} \pi_{\C v,W_i} = 1.
\end{equation}
So, for any measurable set $U\subset \C v$, 
\begin{equation*} 
\sum_{i} |P_i(U)| = |U|,
\end{equation*}
where $P_i:\C v\rightarrow W_i$ is the orthogonal projection. 
\end{theorem}
\begin{proof}
Follows from \autoref{pr:Pv pi}, as $\|v\|^2 = \sum\limits_{i} \left\|P_i v \right\|^2$.
\end{proof}

If $v$ is a quantum state vector and $W$ is an eigenspace of an observable, \eqref{eq:Pv projection factor} shows $\pi_{\C v,W}$ equals the probability given by the Born rule for the corresponding eigenvalue, and \eqref{eq:sum factors} corresponds to the condition of unit total probability. Of course, this still does not mean projection factors are probabilities, specially in EQM, which has no Born rule to begin with. Reaching such interpretation will require a few more steps.

\section{Universes and worlds}\label{sc:universes and worlds}

First, let us describe our structure of universes and worlds, to avoid any confusion with our terminology. 
In the usual Everettian view, there is a single quantum Universe, in a state $\Psi$, with a branch decomposition $\Psi = \sum_{i} \psi_i$ in terms of orthogonal states $\psi_i$ representing quasi-classical worlds. 
We will have, instead, an infinity of identical quantum universes, each with an ill-defined number of distinct quasi-classical worlds. When we take the point of view of one specific universe, we write a capitalized `Universe'.

This structure will result from adding a new postulate to EQM. Some say this violates Everett's purpose, but his ideal of deriving all of quantum theory from its unitary part should not be a shackle. The ultimate goal is to find a coherent theory that agrees with experiments. If it requires augmenting `pure EQM' with new assumptions, and these are reasonable, so be it.
And this is not unprecedented: Deutsch \cite{deutsch1985quantum} has proposed an axiom for EQM which is, in fact, similar to ours. In section \ref{sc:Deutsch} we compare his proposal to ours. 

The \emph{ray} of $\Psi$ is $R_\Psi=\{c\Psi\colon c\in\C^*\}$, i.e. its complex line minus the origin.

\begin{postulate}
There is a continuum infinity of universes, one in each state of a ray $R_\Psi$ (where $\Psi$ is the quantum state of one of them).
\end{postulate}

In section \ref{sc:Physical hypotheses} we discuss why this postulate is reasonable, and whether it might follow from more basic hypotheses.

Since the states of the universes differ only by a complex factor, they are experimentally indistinguishable, i.e. experiments can not tell them apart.
In CQM this followed from the measurement postulate, and in our case it will follow from the linearity of their evolution.
For short, we use the term \emph{identical} for `experimentally indistinguishable'. This agrees with its use in \cite{deutsch1985quantum} or in the expression `identical particles', but differs from the philosophical concept of identity we discuss in section \ref{sc:Identity}.

Decomposing $c\Psi = \sum_{i} c\psi_i$ we get that, as each $c\Psi\in R_\Psi$ represents an actual universe, each $c\psi_i$ corresponds to an actual world in it. So the set of all worlds of type $i$ in all universes corresponds to a ray $R_i = \{c\psi_i: c\in\C^*\}$ of identical worlds, one in each universe.
In other words, the ray of universes $R_\Psi$ decomposes into rays of worlds $R_i$.
Linearity of \Schrodinger's equation implies that all universes evolve and branch in the same way (\autoref{fig:branching_universes}).

\begin{figure}
\centering
\includegraphics[width=.8\linewidth]{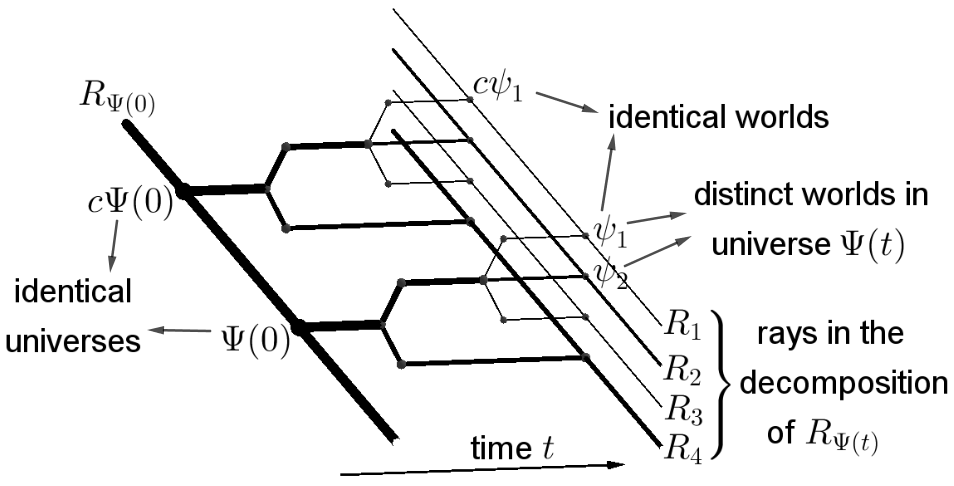}
\caption{Points in ray $R_{\Psi(0)}$ are identical universes, branching into worlds in the same way. The final ray $R_{\Psi(t)}$ is decomposed into rays of worlds. Line thickness indicates Born weight.}
\label{fig:branching_universes}
\end{figure}

\subsection{Relative amounts of worlds}\label{sc:Relative amounts}

We want to quantify, in some sense, the relative amount of worlds of type $i$. This requires a measure, and with continuum infinities of worlds a counting one makes even less sense than before. One might consider counting rays instead of worlds, but Wallace's argument applies: the number of rays is ill-defined, depending on the choice of coarse-graining in the branch decomposition.

Each ray is isometric to an Euclidean plane (minus the origin), and its points are all equivalent. This geometry makes it natural to use the 2-dimensional Lebesgue measure $\left|\cdot\right|$, which is invariant by the group of transformations that preserve the Hilbert space metric (as are positive multiples of $\left|\cdot\right|$, but this does not affect our conclusions).

As a whole ray has infinite measure, we sample a subset of universes $U\subset R_\Psi$ with  $0 < |U| < \infty$. The set of worlds of type $i$ in universes of $U$ is $W_{U,i}=P_i(U)$, where $P_i\colon \C\Psi \rightarrow \C\psi_i$ is the orthogonal projection, and the set of all worlds in all universes of $U$ is a disjoint union $W_U=\bigcup_i W_{U,i}$. Since, by the complex Pythagorean theorem, $|W_U|=\sum_i|W_{U,i}|=|U|$, the fraction $f_{U,i}$ (as measured by $\left|\cdot\right|$) of worlds of type $i$ among all worlds of $U$ is
\begin{equation*} 
	f_{U,i} = \frac{|W_{U,i}|}{|W_U|} = \frac{|P_i(U)|}{|U|} = \pi_{\C\Psi,\C\psi_i}.
\end{equation*}

The factor $\pi_{\C\Psi,\C\psi_i}$ is independent of $U$, so even if the total amounts of worlds in the $R_i$'s can not be compared using their  full measures, which are infinite, their relative amounts  are well defined by the following quantifiers.
\begin{definition}
The \emph{fraction of worlds} of type $i$ (among all worlds in all universes) in $R_\Psi$ is $f_{\Psi,i} = \lim\limits_{r\rightarrow\infty} f_{U_r,i}$, where $U_r=\{\psi\in R_\Psi:\|\psi\|<r\}$. The \emph{percentage of worlds}  of type $i$ in $R_\Psi$ is $f_{\Psi,i}\cdot 100\%$.
The \emph{density of worlds} of type $i$ in $R_\Psi$ is the number $\delta_{\Psi,i}$ such that $|W_{U,i}|=\delta_{\Psi,i}\cdot|W_U|$ for any measurable subset $U\subset R_\Psi$.
\end{definition}

As $f_{\Psi,i}=\delta_{\Psi,i}=\pi_{\C\Psi,\C\psi_i}$, they convey the same information, and which one to use is mainly a matter of linguistic preference, depending on the situation.
When we use expressions like `in almost no world', `in nearly all worlds', etc., we mean these quantifiers are close to 0, to 1, etc. 

By \autoref{pr:Pv pi} and \eqref{eq:Born unnormalized}, 
\[ f_{\Psi,i}=\frac{\|P_i \Psi\|^2}{\|\Psi\|^2}=\frac{\|\psi_i\|^2}{\|\Psi\|^2} = p_{\Psi,i}. \] 
So the probabilities given by the Born rule for quantum experiments equal the relative amounts of worlds with each result. In section \ref{sc:Quantum Fractionalism} we show this is not a coincidence, as these amounts can indeed be perceived as probabilities.

\begin{example}
Let $\Psi=c_1\psi_1+c_2\psi_2$ be the state of an universe after a measurement with two results,  represented by orthogonal world states $\psi_1$ and $\psi_2$. Just to frame our results in the usual way, states are normalized, so $|c_1|^2+|c_2|^2=1$. Any measurable set of universes $U\subset R_\Psi$ decomposes into a disjoint union $W_1\cup W_2$, where each $W_i$ is a set of worlds with result $i$ (\autoref{fig:complexpythagorean}). By \eqref{eq:Pv projection factor} the amounts of worlds of each type are $|W_1|=|c_1|^2\cdot |U|$ and $|W_2|=|c_2|^2\cdot |U|$, so that $|W_1\cup W_2| =|W_1|+|W_2| = |U|$, as in the complex Pythagorean theorem. Thus, in the set of all worlds of all universes, those with each result represent fractions $f_{\Psi,1}=\frac{|W_1|}{|W_1\cup W_2|}=|c_1|^2$ and  $f_{\Psi,2}=\frac{|W_2|}{|W_1\cup W_2|}=|c_2|^2$, with $f_{\Psi,1}+f_{\Psi,2}=1$, corresponding to the Born rule probabilities.
\end{example}

\begin{figure}
\centering
\includegraphics[width=.5\linewidth]{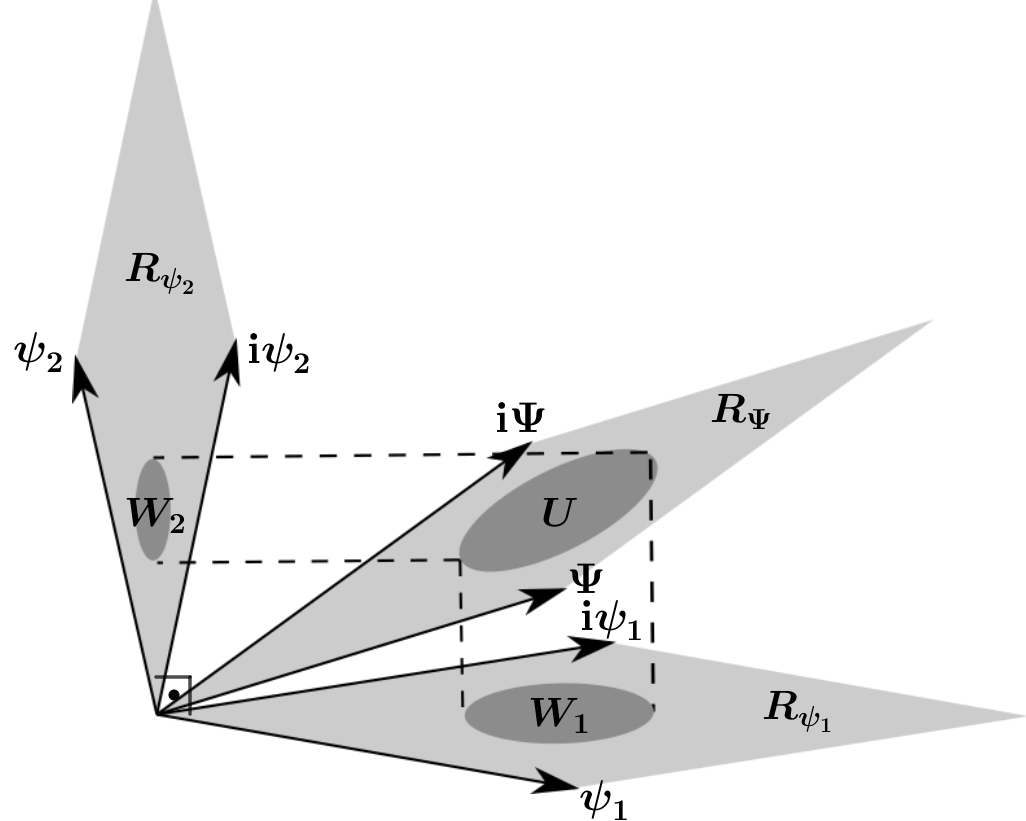}
\caption{Complex Pythagorean Theorem, $|U| = |W_1|+|W_2|$.}
\label{fig:complexpythagorean}
\end{figure}

We note that we have only talked about the $\psi_i$'s as representing worlds because this is the case of most interest. No assumption that these components be quasi-classical or have negligible interference has been used, and our method works with $\Psi$ decomposed in any orthogonal basis
(giving relative amounts of components).
Thus, as we do not rely on a solution to the preferred basis problem, our probabilities can justify using  decoherence to solve it.

\subsection{Non-complex Hilbert spaces}\label{sc:real Hilbert space}

Some authors \cite{Stueckelberg1960a,Finkelstein1962,Adler1995} have proposed real or quaternionic Hilbert spaces for quantum theory.
Before proceeding, we show that a complex space is crucial to get the right probabilities, if EQM and our method are valid.
This may help clarify the role played by the  $\R$-dimension of rays.
We avoid talking of states, worlds or universes here, to stress that our calculations are purely geometric, not requiring a preferred basis of worlds. 

In a real Hilbert space, let $\Psi = \sum_{i} \psi_i$ be a nonzero vector decomposed in components $\psi_i$ with respect to an arbitrary orthogonal basis.
We use the 1-dimensional Lebesgue measure, also denoted by $\left|\cdot\right|$, for subsets of a real ray $R_\Psi=\{c\Psi:c\in\R^*\}$, and take a subset $U\subset R_\Psi$ with  $0 < |U| < \infty$. 
For real lines the projection factor \cite{Mandolesi_Pythagorean} is $\pi_{\R\Psi,\R\psi_i}=\frac{\|\psi_i\|}{\|\Psi\|}$. If $P_i\colon \R\Psi \rightarrow \R\psi_i$ is the orthogonal projection, the set $W_{U,i}=P_i(U)$ of $i$-components (i.e. those in $R_{\psi_i}$) of elements of $U$ measures $|W_{U,i}|=\frac{\|\psi_i\|}{\|\Psi\|}\cdot |U|$, and the disjoint union $W_U=\bigcup_i W_{U,i}$ of all components of all elements of $U$ has $|W_U|= \frac{\sum_i\|\psi_i\|}{\|\Psi\|}\cdot |U|$. Note that $|W_U| \neq |U|$ now. Among all components of all elements of $U$, the $i$-components represent a fraction given by
\begin{equation*} 
	f_{U,i} = \frac{|W_{U,i}|}{|W_U|} = \frac{\|\psi_i\|}{\sum_j\|\psi_j\|}.
\end{equation*}
As before, this does not depend on the choice of $U$, so the fraction $f_{\Psi,i}$ of $i$-components among all components of all elements of $R_\Psi$ is well defined.

If $\Psi=\sum_i c_i\ket{i}$, where $\{\ket{i}\}$ is an orthonormal basis, we have
\[ f_{\Psi,i}=\frac{|c_i|}{\sum_j |c_j|}. \]
Similar calculations show that, since quaternionic rays have $\R$-dimension 4, in a quaternionic Hilbert space we would have
\[ f_{\Psi,i}=\frac{|c_i|^4}{\sum_j |c_j|^4}. \]

By arguments like those in the next section, such fractions would work as probabilities in real or quaternionic versions of EQM. The probability of result $i$ would depend not only on $c_i$ but also on the other coefficients, so this is another example that Gleason's hypothesis is not as natural as it may seem.
The complex case  is special in that $\sum_j |c_j|^2 = \|\Psi\|^2$, while for $p\neq 2$ the value of $\sum_j |c_j|^p$ depends not only on $\Psi$ but also on the decomposition basis.

\section{Quantum Fractionalism}\label{sc:Quantum Fractionalism}

Our last step is to show that relative amounts of worlds can really be interpreted as probabilities.

Let us review the different Everettian points of view in an example.
When an observer $\ket{O}$ measures, in the orthonormal basis $\{\ket{\uparrow},\ket{\downarrow}\}$, an electron spin in state $c_1\ket{\uparrow}+c_2\ket{\downarrow}$, with $|c_1|^2+|c_2|^2=1$ (for simplicity), the process leads to an entangled state
\begin{equation*} 
c_1\ket{\uparrow}\otimes\ket{O_\uparrow}+c_2\ket{\downarrow}\otimes\ket{O_\downarrow}.
\end{equation*}
A naive world counter would claim there are 2 worlds, one with each result, and make wrong predictions. 

Neo-Everettians, on the other hand, say worlds emerge through decoherence, so the final state must take into account the many ways the observer can get entangled with the environment $\ket{E}$, and we have 
\begin{equation}\label{eq:environment}
\sum_i c_{1i}\ket{\uparrow}\otimes\ket{O_{\uparrow i}}\otimes\ket{E_i}+\sum_j c_{2j}\ket{\downarrow}\otimes\ket{O_{\downarrow j}}\otimes\ket{E_j},
\end{equation}with $\sum_i |c_{1i}|^2=|c_1|^2$ and $\sum_j |c_{2j}|^2=|c_2|^2$. 
Coarse-graining separates each sum into worlds having negligible interference, and the number of worlds in each will depend on the chosen fineness of grain. As this number is not an objective feature, it can not be used to say EQM leads to the wrong statistics. 

For us, each complex multiple of \eqref{eq:environment} represents one in a continuum infinity of identical universes, all of which actually exist. As each one is decomposed into worlds, we get an ill-defined number of continua of worlds for each result, with the total fraction (density, if one prefers) of those with  $\ket{\uparrow}$  being $f_\uparrow=|c_1|^2$, while those with $\ket{\downarrow}$ have $f_\downarrow=|c_2|^2=1-f_\uparrow$.

If the experiment is repeated $N$ times, the total fraction of worlds with $n$ ups and $N-n$ downs will be  $\binom{N}{n} \cdot f_\uparrow^n \cdot(1-f_\uparrow)^{N-n}$, a binomial distribution with parameter $p=f_\uparrow$. Thus, for large $N$, the distribution of world fractions, in terms of the frequency $f=\frac{n}{N}$ of ups, will be sharply peaked at $f=f_\uparrow$, with variance $\sigma^2=\frac{f_\uparrow f_\downarrow}{N}$. 
Every possible sequence of results will occur in some world (actually, in continuum infinities of them), but in nearly 99.7\% of them $f$ will deviate from $f_\uparrow$ by at most $3\sigma$.
In section \ref{sc:epistemic problem} we discuss the small percentage of maverick worlds in which results deviate too much.

So, when we use observed frequencies to measure probabilities in a quantum experiment, we are actually measuring the fractions with which worlds branch at each run of the experiment. By \emph{Quantum Fractionalism (or Densitism)} we mean this interpretation of quantum probabilities as being in fact branching fractions or densities. It provides a physically objective concept of probability, at least for quantum experiments, and it might even apply to probabilities in classical settings, if they have quantum origins \cite{Albrecht2014}.

As seen, for multiple trials it provides a realization of all frequentist possibilities, each occurring in the correct fraction of worlds. If we measure $10\,000$ spins in state $\frac{\sqrt{3}}{2}\ket{\uparrow}+\frac{1}{2}\ket{\downarrow}$, every sequence of results will happen in an infinity of worlds, but in 99.7\%  of them the relative frequency of ups will be close to 75\% ($\pm 3\sigma\cong 1.3\%$). Note that the circularity that affected frequentism is not present here, as this 99.7\% is not a likelihood of obtaining $75\% \pm 1.3\%$ of ups (this will certainly happen in infinitely many worlds), but rather a well defined measure of the fraction of worlds in which it will happen.

Quantum Fractionalism works equally well for single quantum trials, with no need to appeal to Bayesian subjective probabilities. 
If I bet a single measurement of $\frac{\sqrt{3}}{2}\ket{\uparrow}+\frac{1}{2}\ket{\downarrow}$ will result up, saying I have a 75\% `chance' of winning has a concrete meaning: after universes branch, I will have won in precisely (disregarding experimental errors) 75\% of all resulting worlds. In section \ref{sc:epistemic problem} we discuss other cases in which our interpretation is an alternative to Bayesianism.

Note that, even if these probabilities have the same values as in CQM, and lead (in most worlds) to the same statistics, their meaning is different. We are used to probabilities referring to alternative random possibilities, but these refer to coexisting certain ones, and must be interpreted for what they are: relative amounts of actual worlds corresponding to each possibility. 

Proponents of EDT have gone to great lengths to ensure rational agents in EQM or CQM would behave in the same way. But knowledge that all results will happen (even if in tiny fractions, in some cases) can and should affect decisions, at least by changing the utility of rewards.
Whatever utility function one might use for decisions in `normal' probabilistic settings, it should be adjusted in the Everettian case to account, for example, for any sadness a winning version of an agent might feel from knowing other versions of herself suffered a loss.

\subsection{Epistemic problem}\label{sc:epistemic problem}

In our interpretation, the half-life of caesium-137 being around 30 years means that, after such period, in nearly all worlds approximately half the atoms in a sample of this material will have decayed. But there will also be an infinity of worlds (albeit representing an extremely low fraction) in which all atoms will have decayed.
Is this reasonable?

In CQM this is a real possibility, only extremely unlikely, so having in EQM a tiny fraction of worlds in which it happens might not be so strange.
Still, perhaps due to the human tendency to equate very low probability with impossibility, it seems less than satisfying that all sorts of unbelievable events should always take place in lots of worlds.

This relates to the \emph{epistemic problem}: if everything that is possible will happen in some worlds, how can observations confirm or refute some quantum hypothesis, or even EQM itself?

This problem is particularly serious for EDT \cite{maudlin2019philosophy}.
Consider a problem of deciding between hypotheses X and Y, with X implying an experimental result will happen with high Born weight, and Y giving it low weight. Regardless of which is correct, there will be worlds with and without the result. If I observe it, I know my branch has high (resp. low) weight if X (resp. Y) is correct, but I can not tell which is the case, since there is no way to measure such weight.
And as EDT ascribes no  post-measurement meaning to this weight, there seems to be no reason why I should give X more credence.

Wallace \cite{Wallace2012} has developed a theory of Everettian inference (assuming the validity of EDT). However, while EDT only says  pre-measurement decisions can neglect future branches of low Born weight, he tries to do the same post-measurement (e.g. in \cite[p.\,202]{Wallace2012}), which is not justifiable in the theory.

In Quantum Fractionalism this is less problematic, as Born weights can be interpreted post-measurement as relative amounts of worlds.
Suppose Alice measures a spin in state $\ket{\uparrow}+\ket{\downarrow}$. If her result is $\uparrow$ (labeled as $A_\uparrow$ in \autoref{fig:AliceBob}) she prepares a large number of spins in state $\ket{\uparrow}$, otherwise ($A_\downarrow$) she prepares them in state $\ket{\uparrow}+\ket{\downarrow}$. By measuring these spins, one at a time, Bob is to decide whether he is in an $A_\uparrow$ or $A_\downarrow$ world. Once he gets a $\downarrow$ result he knows for sure (disregarding experimental errors) he is in $A_\downarrow$, so let us focus on what he can infer if he gets $\uparrow$ every time.

\begin{figure}[h!]
\centering
\includegraphics[width=0.6\linewidth]{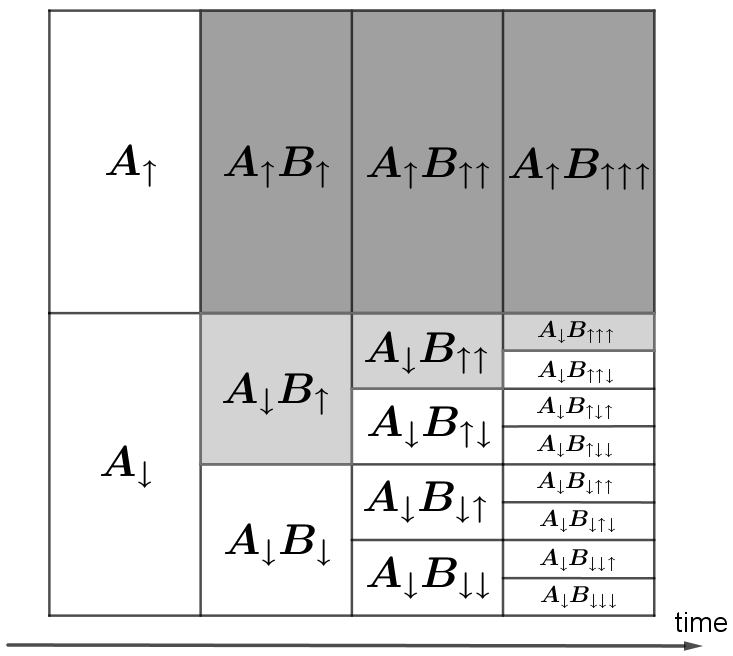}
\caption{World distribution of Alice and Bob's results. Columns show the situation after each measurement, with areas proportional to the fractions of worlds with each sequence of results. In shaded worlds Bob has gotten only $\uparrow$'s, and in the darker ones Alice got $\uparrow$ too.}
\label{fig:AliceBob}
\end{figure}

His first measurement results $\uparrow$ (labeled as $B_\uparrow$) in 75\% of all worlds, with $2/3$ of these in $A_\uparrow$, so at this time he should ascribe a $66.7\%$ credence to being in $A_\uparrow$. With another $\uparrow$ this should rise to $80\%$, as $5/8$ of all worlds have $B_{\uparrow\uparrow}$, and $A_\uparrow$ remains as $1/2$ of all worlds. After a sequence of $n$ $\uparrow$'s, his credence will be $\frac{2^n}{2^n+1}$, rapidly approaching 100\%. With 10 $\uparrow$'s he should be 99.9\% sure of being in $A_\uparrow$, with such credence having an objective meaning: 99.9\% of all worlds in which he got such sequence are $A_\uparrow$.

Still, in $1/2^{n+1}$ of all worlds Bob will be extremely confident of being in $A_\uparrow$, but he will be wrong. Even if these worlds form only a tiny fraction of the set of all worlds, they do not seem any less real than the rest. In CQM there is a chance experimental results are misleading, but we have learned to treat it as zero if it is small enough.
It may be harder to accept that EQM guarantees we will be misled in some worlds.

There might be a way to rid EQM of such maverick worlds, at least in some cases.
As suggested in \cite{Mandolesi2017}, branches with weight many orders of magnitude lower than similar ones might not form stable quasi-classical worlds, as interference from the larger ones precludes macroscopic causality in them. 
Such idea still requires more work to be made precise, but if correct there might be no $\ket{A_\downarrow B_{\uparrow\uparrow\cdots\uparrow}}$ world for $n$ large enough.
So, in the quantum case, extremely low probability (fraction, to be precise) might indeed equal impossibility.

Even if these worlds do exist, this does not mean experimental results can not be trusted anymore, bringing science to a halt. We just have to interpret correctly what they mean, having in mind that trust is not absolute, but a matter of degree. 
In a situation as the one above, reaching 99.9\% credence in an hypothesis means that in 99.9\% of all admissible worlds (those with a supporting set of experimental results) this is the right conclusion, but there is a $0.1\%$ chance we might be in a branch where experiments have deceived us (i.e. this will surely happen in $1/1000$ of all admissible worlds). If this chance is not small enough, we can reduce it with more experiments, but it will never be $0$, and we may have to live with the fact that there is a non null chance  we have been misled. 
We may not like the idea of following the scientific method and still being fooled like this, but the Universe is under no obligation of ensuring we are always guided towards the truth. Maybe we should feel lucky enough that this will happen in the vast majority of worlds. 

Let us consider now the problem of confirming EQM itself. Can a theory predicting everything will happen be tested? Should we take any experimental result as evidence in its favor, or should we reject it for not being falsifiable?
Again, we must have in mind that, in general, experimental evidence does not prove or disprove a hypothesis or theory in absolute terms, it only increases or decreases our confidence in them. 

Suppose we measure a large number of spins in state $\frac{\sqrt{3}}{2}\ket{\uparrow}+\frac{1}{2}\ket{\downarrow}$, and based on the results we are to decide between EQM, CQM and Naive Branch Counting (NBC), a branching theory in which each measurement of this state produces 2 worlds, one for each result. Let us consider what we can conclude if we obtain a sequence of results having, approximately:
\begin{enumerate}[a)]
\item 0\% of $\uparrow$'s and 100\% of $\downarrow$'s;
\item 50\% of $\uparrow$'s and 50\% of $\downarrow$'s;
\item 75\% of $\uparrow$'s and 25\% of $\downarrow$'s.
\end{enumerate}

These are real possibilities in all theories.
In EQM, (a) and (b) happen in tiny fractions of worlds, while (c) happens in most of them. In CQM, (a) and (b) are unlikely, and (c) is expected to happen with high probability. And in NBC (a) and (c) happen in few worlds, with (b) happening in most.

Observation of (a) should cast doubt on all theories, specially CQM, but it does not disprove any of them. In case (b) we should give more credence to NBC than to EQM or CQM. And (c) would tell us to doubt NBC, and trust more EQM or CQM. 
But again, there is no guarantee this strategy will lead us to the right conclusion. If CQM is correct, we do not have to worry about (a) or (b), which are unlikely, but in the odd chance they happen, this might lead us to wrongfully reject the theory. If either EQM or NBC is right, there will certainly be worlds in which we are led to the wrong conclusion, though this will happen in a tiny minority of them.
And perhaps the best we can hope for is to be right in nearly all worlds.

If it is any consolation, and EQM is correct, physicists in branches (a) and (b) might at least have their illusions about CQM dispelled. In (c), any preference between EQM and CQM has to be based on their logical coherence.

These examples show that Quantum Fractionalism can do much of the work usually reserved for Bayesianism, at least in situations regarding quantum measurements, and it provides a more concrete interpretation of what the probabilities mean. Of course, this does not mean that all subjective probabilities can be reinterpreted in this way.
If I say I am 80\% convinced EQM is correct, this does not mean the theory only holds in 80\% of all branches.

\section{Deutsch's continuum of universes}\label{sc:Deutsch}

Before EDT, Deutsch \cite{deutsch1985quantum,Deutsch1997} had also tried to solve the probability problem by postulating a continuum of universes. In this section we compare our proposal to his, to dispel any  misunderstandings.
Terminologies differ, so we write world$_{\scriptscriptstyle{D}}$, branch$_{\scriptscriptstyle{D}}$ or universe$_{\scriptscriptstyle{D}}$ when these terms are used in his sense. 

For us, there is an infinity of identical universes, each consisting of a quantum superposition of different worlds or branches. For Deutsch, there is a single world$_{\scriptscriptstyle{D}}$ (or \emph{multiverse}), which is an infinite set of universes$_{\scriptscriptstyle{D}}$, with branches$_{\scriptscriptstyle{D}}$ being subsets of identical universes$_{\scriptscriptstyle{D}}$.

He defines the world$_{\scriptscriptstyle{D}}$ as everything that exists, postulating that it consists of a continuum infinity of universes$_{\scriptscriptstyle{D}}$, and that this set has a measure (for simplicity, let the total measure be $1$). 
The state of the world$_{\scriptscriptstyle{D}}$ is described by a normalized $\Psi=\sum_i c_i\psi_i$, where the $\psi_i$'s are states of a preferred orthonormal basis (Deutsch's \emph{interpretation basis}).
The meaning attributed to $\Psi$ is that the set of universes$_{\scriptscriptstyle{D}}$ is divided into disjoint subsets (branches$_{\scriptscriptstyle{D}}$), each with measure $|c_i|^2$ and consisting of identical universes$_{\scriptscriptstyle{D}}$ in state $\psi_i$ (\autoref{fig:Deutsch}).

\begin{figure}
\centering
\begin{subfigure}[t]{0.48\textwidth}
\includegraphics[width=\textwidth]{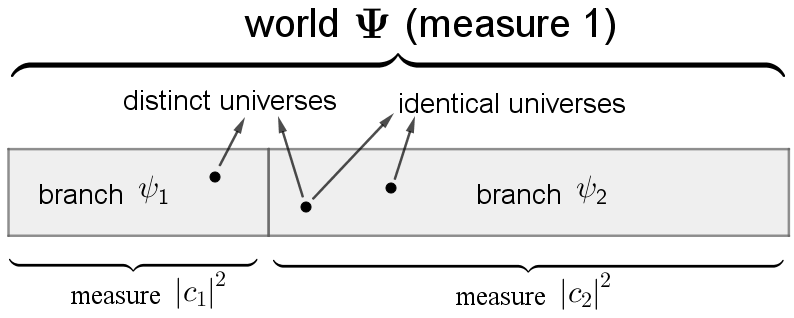}
\caption{Deutsch's model. The world in state $\Psi$ is a set of universes, with measure $1$, divided into  subsets of identical universes, the branches $\psi_1$ and $\psi_2$, with measures $|c_1|^2$ and $|c_2|^2$.}
\label{fig:Deutsch}
\end{subfigure}
\hspace{.02\textwidth}
\begin{subfigure}[t]{0.48\textwidth}
\includegraphics[width=\textwidth]{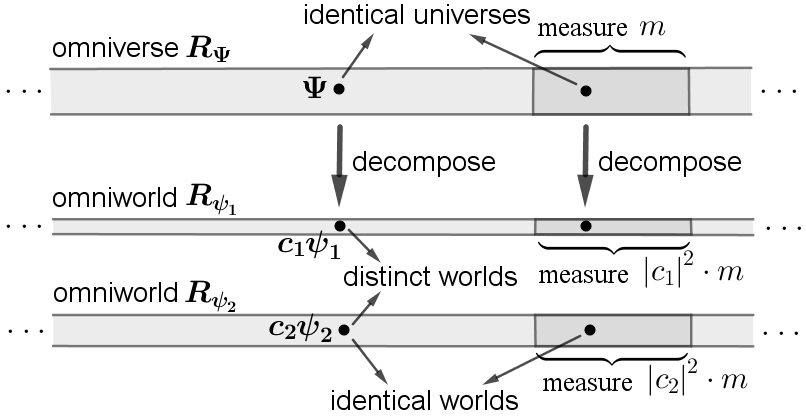}
\caption{Our model. Omniverse is a ray $R_\Psi$ of identical universes, of infinite measure. An universe $\Psi$ decomposes as worlds $c_i\psi_i$, and $R_\Psi$ decomposes as rays $R_{\psi_i}$. Subset of measure $m$ in $R_\Psi$ decomposes into subsets of measures $|c_i|^2\cdot m$ in the $R_{\psi_i}$'s.}
\label{fig:ours}
\end{subfigure}
\caption{Interpretation of $\Psi=c_1\psi_1+c_2\psi_2$ in both models.}
\label{fig:models}
\end{figure}

As Deutsch puts it, this `builds in' the probabilistic interpretation into EQM. To get branches$_{\scriptscriptstyle{D}}$ with the desired probabilities, he simply postulates them to be sets of identical universes$_{\scriptscriptstyle{D}}$, with appropriate sizes.
In a process like \eqref{eq:device branches}, the world$_{\scriptscriptstyle{D}}$ has initially a single branch$_{\scriptscriptstyle{D}}$ (for simplicity), with a continuum infinity of identical universes$_{\scriptscriptstyle{D}}$, all of which are supposed to actually exist. Then it evolves into another ensemble of universes$_{\scriptscriptstyle{D}}$, not all identical, in which those of type $i$ form a subset of measure $|c_i|^2$. 

This model reflects well experimental results, but it is not much of an improvement over CQM. The ensemble is supposed to evolve deterministically according to \Schrodinger's equation, but it is not clear what happens to individual universes$_{\scriptscriptstyle{D}}$. 
One view is that universes$_{\scriptscriptstyle{D}}$ that were identical diverge, evolving into different ones, in distinct branches$_{\scriptscriptstyle{D}}$. But as there is no explanation of how this happens, or what determines which universes$_{\scriptscriptstyle{D}}$ will end up in each branch$_{\scriptscriptstyle{D}}$, it seems to bring probabilistic collapses back into EQM, or some sort of hidden-variables.
Another possible view is that identical universes$_{\scriptscriptstyle{D}}$ do evolve in the same way, with identical new sets of universes$_{\scriptscriptstyle{D}}$ sprouting from each of them, divided in branches$_{\scriptscriptstyle{D}}$ with the appropriate measures. Again, it is unclear how this would happen.

Another limitation of Deutsch's approach is that it needs the interpretation basis to tell which branches$_{\scriptscriptstyle{D}}$ are present in the ensemble, so even if it worked there would be a circularity in using its probabilities to solve the preferred basis problem via decoherence.

For us, the set of everything that exists (call it \emph{omniverse}) is also a continuum of universes, but now individual universes have their own quantum states. But these differ only by complex factors, so the universes are identical, and there is a correspondence between the omniverse and a ray $R_\Psi$ (\autoref{fig:ours}). 
Evolution via \Schrodinger's equation takes rays to rays, so all universes remain identical at all times.
Our continuum also has a measure, but it is naturally inherited from the Hilbert space metric through the correspondence with $R_\Psi$.

While Deutsch's branches$_{\scriptscriptstyle{D}}$ were subsets of universes$_{\scriptscriptstyle{D}}$ which remained identical after a measurement, our branches or worlds are a substructure of the universes, forming a new level in the structure of the omniverse. Each world corresponds to a component of the state of one of the universes. As all universes are identical, they have identical world structures, and the set of all worlds of type $i$ (call it an \emph{omniworld}) in all universes corresponds to a ray $R_{\psi_i}$.

In our omniverse composed of a quantum superposition of omniworlds, the complex Hilbert space geometry decomposes each subset of measure $m$ in $R_\Psi$ into subsets of measures $|c_i|^2\cdot m$ in the $R_{\psi_i}$'s. So each kind of world is present in the correct ratio to give the Born rule probabilities, without us having to postulate the measures.

And, as our calculations do not really require a decomposition into worlds, working with any decomposition basis, there is no circularity in using our probabilities to solve the preferred basis problem.

\section{Why a ray of universes?}\label{sc:Physical hypotheses}

Many physicists dislike EQM for its worlds whose existence we can not verify (unless, perhaps, we can detect tiny interferences between them), and they will frown at the idea of an infinity of universes which can not be detected in any way, as they do not  interfere (although we may consider the Born rule as indirect evidence that they exist).
So, our postulate will be contested, and the obvious question is not only why there is a continuum of universes, but why specifically a ray, and not some other subset of Hilbert space. 

The easy answer is that, so far, the only coherent explanation we have for the Born rule in EQM requires this particular configuration of universes. As any theory, ours depends on postulates. If its internal logic is sound, and it agrees with experiments, should this not be good enough?

If this argument sounds a little too Copenhagen-esque, that is because it is. 
In fact, many old arguments in defense of CQM can be recycled for our use, with the advantage that EQM, with its probability and preferred basis problems solved, is more theoretically consistent than CQM ever was. Our postulate is simpler than the measurement one, does not depend on ill-defined terms like `measurement', and may lead to a better understanding of both quantum theory and probability. This should be reason enough to consider its validity. 

Still, history shows this kind of reasoning may serve us well for a while (even a long one), but eventually we must overcome it, as lack of answers for basic questions sooner or later starts to hamper progress.
So, in this section, we try to do a little better, speculating whether our postulate can be obtained from more fundamental hypotheses.

We will examine our postulate from two perspectives. First, assuming its validity, we ask in sections \ref{sc:Identity} and \ref{sc:Ultradeterminism} what insights it might bring us.
Our conclusion is that the following ideas should be called into question, and in section \ref{sc:ray of Universes} we show our postulate might follow from their negation:
\begin{enumerate}
\item Complex multiples of a quantum state vector are different mathematical representations of the same physical state.

This idea stems from the fact that all observable properties of $\Psi$ hold for any $c\Psi\in R_\Psi$. Many physicists see this as a mathematical redundancy with no physical implications, and some eliminate it using  $R_\Psi$  (or, equivalently, the density operator of a normalized pure $\Psi$) to describe a quantum state.
For them, the fundamental quantum space is not a Hilbert space $H$, but a \emph{projective Hilbert space} $\PP(H) = \{$rays of $H\}$, whose geometry provides nice interpretations for concepts like the Berry phase \cite{Bengtsson2017}. But as linear combinations of rays are undefined, the superposition principle is replaced by a decomposition one \cite{Boya1989}: instead of writing $\Psi=\sum_i \psi_i$ one says $R_\Psi$ decomposes in the $R_{\psi_i}$'s, with probabilities given by squared cosines of angular distances between them. This makes even proponents of such approach often revert to adding quantum states.

\item The existence of the Universe may be the result of a random event.

This possibility appears in theories of quantum cosmology \cite{Carr2007}, in which the Universe emerges from  quantum vacuum fluctuations. 
\end{enumerate}

\subsection{Identity and indistinguishability}\label{sc:Identity}

Instead of using `identical' as shorthand for `experimentally indistinguishable', in this section we use it in the philosophical sense of `one and the same', or `numerically the same', with `different' being its opposite. 

As our postulate requires an universe for each state in a ray, this suggests the existence of all these states is not redundant, playing a fundamental role in the theory. 
So we are led to the following conjecture:
\begin{itemize}
\item[C1)] Different elements of a ray $R_\Psi$ represent different, but experimentally indistinguishable, physical states.
\end{itemize}
We shall even speculate that no physical principle, even those we do not know yet, makes any distinction between states $\Psi$ and $c\Psi$ ($c\in\C^*$). 

But can indistinguishable states be considered different?
Yes, since difference is a more fundamental relation than distinguishability, which depends on the means available. 
Even if we can not prove the states represented by $\Psi$ and $c\Psi$ are different, this does not imply they are the same.
Their norms and global phases may be real physical features, which we simply can not measure. Physical laws may forbid experimental access to such properties, but this does not mean they do not exist. Relativity keeps us from observing parts of the Universe that are too far away, and still they are assumed to exist.

The philosophical debate about what it means to say things $x$ and $y$ are identical is an old one. Leibniz's principle of Identity of Indiscernibles claims that for them to be different one must possess a predicate the other does not. So, if $x$ and $y$ are qualitatively the same, they are numerically the same.
But it is not so simple: the principle is a tautology unless we exclude certain predicates, like `is identical to $x$', and much of the debate revolves around which ones to allow. Black \cite{black1952identity} has proposed  as counterexample a symmetric universe with only 2 perfect spheres of same radius, as there is no way to distinguish them, and yet they are not one. Quantum theory, with its fermions and bosons, has complicated matters: electrons being fundamentally indistinguishable does not mean they are one single particle present in multiple locations. Reviews of ideas surrounding this problem, and lists of references, can be found in \cite{Forrest2016,French2019}.

For us, the problem boils down to whether predicates like `has norm 2' can hold for quantum states, even if we can not test their validity.  
In CQM this is underdetermined, as in its Hilbert space formulation states have norms, but $\PP(H)$ is an equally valid model in which they do not (or have norm 1). In EQM, the Born rule may be indirect evidence that quantum states do have norms and global phases, even if these are not observable, for we need $\Psi$ and $c\Psi$ ($c\neq 1$) to be numerically distinct in order to get the right probabilities.
If each ray is reduced to a single point in $\PP(H)$, we get a counting measure instead.

%

\subsection{Deterministic existence}\label{sc:Ultradeterminism}

Discarding CQM, all fundamental laws of Physics are deterministic, so in EQM there seems to be no room for randomness.
Even the idea that the Universe appeared due to random quantum fluctuations should be reinterpreted in terms of an initial state $\Psi_0$ with a superposition of all possible beginnings, which would evolve into a superposition of all possible worlds. 

In fact, being governed by \Schrodinger's equation, the evolution of the Universe is not only deterministic but also backwards deterministic (the state at any time determines the states at all previous times) and linear (so all universes in a ray evolve in tandem).
Thus, if our postulate is valid and at time $t$ there is a ray $R_t$ of universes, tracing their evolution backwards to the Big Bang (assuming it happened, and something like EQM holds in its extreme conditions) we infer that all universes started in the same ray $R_0$.
This is hardly fortuitous, so $R_0$ must contain some special \emph{Big Bang or seed state} $\Psi_0$. 

It seems reasonable to assume $\Psi_0$ not only can, but actually must, give rise to an universe, or we would need a probabilistic process to determine whether or not this happens. 
As Everettians say, in EQM anything that can happen, will happen.
Also, whatever makes $\Psi_0$ special must be independent of the universe it brings forth, preceding it (whatever this means) as some physical principle. 

We formalize these ideas as another conjecture: 
\begin{itemize}
\item[C2)] The initial state $\Psi_0$ of the Universe was not random. Instead, some physical principle determines an universe must sprout from $\Psi_0$. 
\end{itemize}

This may seem too unorthodox. In Physics, initial conditions are usually set at will, and determinism only guides evolution after such starting point. But applying this to the Universe is problematic: unless we veer into religion, or accept that the initial conditions of the Universe may simply have no cause (sending science into a dead end), there is no way these could have been set up but through some physical principle. Of course, this only pushes back the problem to the cause of physical principles, but nonetheless our job is to push it back as far as we can.

If there was no Big Bang and the universes always existed, their trajectories in Hilbert space are curves with no starting point, which by our postulate are passing through all points in a ray $R_\Psi$. It is again hard to imagine this is an accident, so we are led to a more general version of C2 (which can also be adapted to accommodate a `block universe' relativistic perspective):
\begin{itemize}
\item[C2')] The existence of the Universe is determined by Physics, i.e. there is a physical reason for its existence in a specific state at a given time. 
\end{itemize}

\subsection{Lo and behold, a ray of universes}\label{sc:ray of Universes}

Now we work in the other direction, assuming our conjectures are valid. 

From C2 we get a seed state $\Psi_0$ and some physical principle making it give rise to an universe. By C1, any other $c\Psi_0\in R_{\Psi_0}$ is a different state, which no physical law distinguishes from $\Psi_0$. Hence the same physical principle applies to it, giving rise to a different, but indistinguishable, universe.
Thus $R_{\Psi_0}$ is a whole ray of baby universes.
Determinism and linearity make them evolve in synchrony: if at time $t$ an universe is in state $\Psi(t)$, the states of all universes will form the ray $R_{\Psi(t)}$, as in our postulate.
A similar argument holds with C2', leading to the same conclusion.

Granted, there are big holes in all this.
Our present knowledge does not allow us to be certain about whether the universe has a quantum state; what it means to say it exists; why it exists; if there can be more than one; whether it had a beginning; what exactly are time and evolution; and whether there is any sense in saying the universes evolve in synchrony.

But these difficulties are not specific to EQM or to our approach, and if Physics is ever to address such questions we must advance hypotheses and see if they lead to reasonable conclusions.
A point in favor of our conjectures, besides their simplicity, is that they might explain our postulate, which leads to the Born rule. And, perhaps, the fresh perspective they provide might even bring forth new insights into those questions.

Some see the resilience of the measurement postulate as an indication that it can not be so wrong, but it may be the opposite: it threw physicists so off track that it has been hard to find the way back.
By contrast, it was easier to look beyond our postulate and find more fundamental assumptions that might explain it, even if the path connecting them is not so clear yet. We take this as an indication that we are on the right track.

\section{Final remarks}\label{sc:Conclusion}

Our postulate allows quantum probabilities to be interpreted, in EQM, as relative amounts of worlds. This solution of the probability problem does not require a branch decomposition in quasi-classical worlds, working with any orthogonal basis. Hence there is no circularity in using our probabilities to justify the decoherence based solution of the preferred basis problem. Thus, the main objections to EQM are eliminated. 

Popper's falsifiability principle has also been invoked against EQM, with the claim that it makes no testable predictions differing from CQM. But if they make the same predictions, any test of one is a test of the other, and historical antecedence does not make one more falsifiable than the other. 
If CQM had been proposed after EQM, quite possibly it would be the one accused of being unfalsifiable.
Anyway, as more sophisticated experiments test quantum effects at ever larger scales, and `pure EQM' is supplemented by new ideas, it is possible that in the future it can be tested against CQM \cite{Deutsch1986,Plaga1997,Page1999}.

Even if EQM is in fact experimentally equivalent to CQM, its logical structure, with those problems solved, is superior. Among its advantages, we have: 
it has no ill-defined concept of measurement causing drastic changes in quantum behavior for no clear reason; 
observers have no special status; 
there is no mysterious colapse of the wavefunction; 
entanglement involves no `spooky action at a distance';
its limit of applicability is not restricted by some unknown boundary; 
it encompasses classical mechanics.
And it may have greater explanatory power, possibly helping clarify even the nature of classical probabilities.

Critics who wield Occam's razor against EQM will like even less the idea of a continuum of identical universes, each with lots of distinct worlds. But objections against unnecessarily complex explanations only apply if a simpler theory explains the same facts, while also being theoretically sound. Simply claiming a theory with infinitely many universes is needlessly complex is like attacking modern astronomy for requiring billions of galaxies to explain observations. To reject astronomy one would have to come up with a better explanation for the same observations, and the same standard should hold for EQM: to reject it one must present an alternative which explains the same experiments, and whose internal logic is in some sense better.
Besides CQM and EQM, many other interpretations of quantum theory have been proposed, like hidden-variables, Bohmian mechanics, the Ghirardi-Rimini-Weber theory, consistent and decoherent histories, etc., all of which face various problems \cite{Auletta2000,Home1997,Wheeler2014}.

Pragmatists tend to be satisfied with CQM for its correct description of experiments, dismissing EQM for its untestable elements and for bringing no obvious experimental advantages that justify switching from CQM. But as EQM allows us to get the same results in a more coherent theoretical framework, it provides clearer ways to analyze quantum problems. And if all one cares about are measurable results, it is irrelevant whether there really are many worlds or not. From a pragmatic perspective, one is free to consider them as just mathematical artifacts of a formalism which, nevertheless, churns out correct results in a simpler way.

As discussed, our postulate suggests conjectures which require some common physical ideas to be reexamined.
But the history of quantum theory is full of ideas (like Bohr's atomic model or Planck's quanta of energy) which challenged the physical knowledge of the time. Some even turned out to be wrong, but provided the seed for other developments, until there were enough new ideas to form a whole new paradigm.
Granted, those ideas were only accepted as they led to testable predictions.
But we can again blame historical tardiness: if EQM and our postulate had appeared before CQM, their predictions would have been confirmed by all experiments. 
Anyway, that problems of CQM remain a century later suggests another paradigm change may be necessary, and new ideas should deserve serious consideration, even if they go against old conceptions. Even if they turn out to be not quite right, they might still provide  clues of what to look for.

\section*{Acknowledgments}

The author wishes to thank the anonymous referees for their helpful criticisms and suggestions.


\providecommand{\bysame}{\leavevmode\hbox to3em{\hrulefill}\thinspace}
\providecommand{\MR}{\relax\ifhmode\unskip\space\fi MR }
\providecommand{\MRhref}[2]{%
  \href{http://www.ams.org/mathscinet-getitem?mr=#1}{#2}
}
\providecommand{\href}[2]{#2}

\end{document}